\documentclass[12pt]{article}
\usepackage{slashbox}
\usepackage{multirow}
\usepackage[T1]{fontenc}
\usepackage[cp1250]{inputenc}

\usepackage{graphicx}
\usepackage{color}
\usepackage{enumerate}
\usepackage{rotating}
\usepackage{tikz}
\usepackage{algpseudocode}
\usepackage{algorithm}

\usepackage{amsmath, amsthm, amsfonts, amssymb}

\newtheorem{theorem}{Theorem}[section]

\newtheorem{corollary}[theorem]{Corollary}

\newtheorem{lemma}[theorem]{Lemma}

\begin{document}
\markboth{ }{}
\title{\bf Scheduling of unit-length jobs with bipartite incompatibility graphs on four uniform machines\footnote{Project has been 
partially supported by Narodowe Centrum Nauki under contract 
DEC-2011/02/A/ST6/00201}}
\date{}
\author{Hanna Furma\'nczyk\footnote{Institute of Informatics,\ University of Gda\'nsk,\ Wita Stwosza 57, \ 80-309 Gda\'nsk, \ Poland. \ e-mail: hanna@inf.ug.edu.pl}, 
\ Marek Kubale \footnote{Department of Algorithms and System Modelling,\ Gda\'nsk University of Technology,\ Naruto\-wi\-cza 11/12, \ 80-233 Gdańsk, \ Poland. \ e-mail: kubale@eti.pg.gda.pl}}

\markboth{H. Furma\'nczyk, M. Kubale}{Scheduling of unit-length jobs on four uniform machines}

\maketitle

\begin{abstract}
In the paper we consider the problem of scheduling $n$ identical jobs on 4 uniform machines with speeds $s_1 \geq s_2 \geq s_3 \geq s_4,$ respectively. Our aim is to find a schedule with 
a minimum possible length. We assume that jobs are subject to some kind of mutual exclusion constraints modeled by a bipartite incompatibility graph of degree $\Delta$, where two incompatible 
jobs cannot be processed on the same machine. We show that the problem is NP-hard even if $s_1=s_2=s_3$. If, however, $\Delta \leq 4$ and $s_1 \geq 12 s_2$, $s_2=s_3=s_4$, then the problem 
can be solved to optimality in time $O(n^{1.5})$. The same algorithm returns a solution of value at most 2 times optimal provided that $s_1 \geq 2s_2$. Finally, we study the case 
$s_1 \geq s_2 \geq s_3=s_4$ and give an $O(n^{1.5})$-time $32/15$-approximation algorithm in all such situations.
\end{abstract}


\section{Introduction}
Imagine you have to arrange a dinner for 40 people and you have at your disposal 4 round tables with different numbers of seats (not greater than 16). 
You know that among your guests there are vegetarians and non-vegetarians. Moreover, each vegetarian is in bad relations with at most 4 non-vegetarians and vice-versa. 
Your task is to assign the people to the tables in such a way that no two of them being in bad relations seat at the same table. In the paper we show how to solve 
this and related problems.

Our problem can be expressed as the following scheduling problem. Suppose we have $n$ identical jobs $j_1,\ldots,j_n$, so we assume that they all have unit execution 
times, in symbols $p_i=1$, to be processed on four non-identical machines $M_1, M_2, M_3,$ and $M_4$. These machines run at different speeds $s_1 \geq s_2 \geq s_3 \geq s_4$, respectively. 
However, they are uniform in the sense that if a job is executed on a machine $M_i$, it takes $1/s_i$ time units to be completed. It refers to the situation where the machines are of 
different generations, e.g. old and slow, new and fast, etc.

Our scheduling problem would be trivial if all the jobs were compatible. Therefore we assume that some pairs of jobs cannot be processed on the same machine due to some technological 
constraints. More precisely, we assume that each job is in conflict with at least 1 and at most 4 other jobs. Moreover, we assume that the underlying incompatibility graph $G$, whose 
vertices are jobs and edges correspond to pairs of jobs being in conflict, is a bipartite graph (without isolated vertices). For example, all graphs in our figures are bipartite. Notice that 
two jobs being in conflict may be executed in intersecting time intervals. A load of $k$ jobs on a machine $M_i$ requires the processing time $k/s_i$, and all the jobs are ready for 
processing at the same time. Alternatively, if a load on $M_i$ is not given explicitly, we are using the notation $C(M_i)$ to mean the schedule length on the machine $M_i$. By definition, each 
load forms 
an independent set (color) in $G$. Therefore, in what follows we will be using the terms job/vertex and load/color/independent set interchangeably. Since all the tasks have to be executed, 
the problem is to find a 4-coloring, i.e. a decomposition of $G$ into 4 independent sets $I_1, I_2, I_3,$ and $I_4$ such that the schedule length $C_{\max} = \max\{|I_i|/s_i: i =1,\ldots,4\}$ 
is minimized, in symbols $Q4|p_i=1,G=bipartite| C_{\max}$.

There are several papers devoted to chromatic scheduling in the presence of mutual exclusion constraints. Boudhar in \cite{boudhar,boudhar3} studied the problem of batch scheduling with 
complements of bipartite and split graphs, respectively. Finke et al. \cite{finke} considered this problem with complements of interval graphs. Other models of batch scheduling with 
incompatibility constraints were studied in \cite{demange,werra}. Our problem can also be viewed as a particular variant of scheduling with conflicts \cite{even}. In all the papers the 
authors assumed identical parallel machines. However, to the best of our knowledge little work has been done on scheduling problems with uniform machines involved (cf. \cite{furm,li}).

The rest of this paper is organized as follows. In Section \ref{np} we show that  the general problem is NP-hard even if $s_1=s_2=s_3$. In Section \ref{alg} we show that if $s_1 \geq 12s_2,$ $s_2=s_3=s_4$, 
then the problem can be solved to optimality in time $O(n^{1.5})$ provided that the degree of $G$ is $\Delta \leq 4$. The same algorithm returns a solution of value at most 2 times optimal 
provided that $s_1 \geq 2s_2$, $s_2=s_3=s_4$. In Section 4 we study the case $s_1 \geq s_2  \geq s_3=s_4$ and give an $O(n^{1.5})$-time $32/15$-approximation algorithm in 
all such cases. Finally, we discuss possible extensions of our model to more than four machines. 

\section{NP-completeness proof}\label{np}
We begin with introducing a few basic notions concerning graph coloring. Given graph $G=(V,E)$, a $k$-coloring of $G$ is a mapping $c: V \rightarrow \{1, \ldots ,k\}$ such that for all edges 
$\{u,v\} \in E$ we have $c(u) \neq c(v)$. The smallest $k$ for which $G$ is $k$-colorable is called the \emph{chromatic number} of $G$ and denoted $\chi(G)$. A graph $G=(V,E)$ is said to be 
\emph{equitably $k$-colorable} if and only if its vertex set can be partitioned into independent sets $V_1, \ldots,V_k \subset V$, possibly empty, such that $||V_i| - |V_j|| \leq 1$ for all $i,j =1,\ldots,k$. 
The smallest $k$ for which $G$ admits such a coloring is called the \emph{equitable chromatic number} of $G$ and denoted $\chi_=(G)$. Graph $G$ has a \emph{semi-equitable $k$-coloring} 
$(k \geq 3)$, if there exists a partition of its vertices into independent sets $V_1,\ldots,V_k \subset V$ such that one of these subsets, say $V_i$, is of size 
$\not \in \{\lfloor n/k \rfloor, \lceil n/k \rceil \}$, and the remaining subgraph $G - V_i$ is equitably $(k-1)$-colorable. In the following we will say that graph $G$ has 
$(V_1,\ldots ,V_k)$-coloring to express explicitly a partition of $V$ into $k$ independent sets. If, however, only the cardinalities of color classes are important, we will use the 
notation $[|V_1|,\ldots,|V_k|]$. For example, the graph in Fig.~\ref{fig1} has one equitable coloring of type $[3,2,2,2]$, one semi-equitable coloring of type $[6,1,1,1]$ and several other types of 
colorings. 

Let us recall some basic facts concerning the colorability of bipartite graphs. First of all, for any bipartite graph $G$ we have $\chi(G) = 2$. Such a 2-coloring can be obtained in time 
proportional to the size of $G$ while traversing it in a \texttt{DFS} order. Moreover, Chen and Yen \cite{chen} proved that any bipartite graph $G$ with $\Delta \geq 2$ is equitably 
$\Delta$-colorable in linear time if and only if $G$ is different from a complete bipartite graph $K_{2q+1,2q+1}$ for all $q \geq 1$. 

The maximal size of an independent set in $G$ is called the \emph{independence number} of $G$ and denoted $\alpha(G)$. Since $\alpha(G)\chi(G)\geq n$, we have a lower bound 
$\alpha(G)\geq n/\chi(G)$ on it. On the other hand, the maximal gap between the sizes of independent sets in a bipartite graph is for $K_{1,\Delta}$. This follows that 
$\alpha(G) \leq n\Delta/(\Delta+1)$. Since in our case $\chi(G)=2$ and $\Delta \leq 4$, we have
\begin{equation}
n/2 \leq \alpha(G) \leq 4n/5.\label{ine}
\end{equation}
Note that an independent set of size $\alpha(G)$ can be computed in $O(n^{1.5})$ time by finding a maximum matching in $G$ (see Hopcroft and Karp \cite{hop}), since one of the two endpoints of each edge in the maximum 
matching belongs to the complement of maximum independent set in $G$.

In the following we will need the \texttt{Partition Into Bounded Independent Sets} problem, which is defined as follows: Given a graph $G=(V,E)$ and positive integers $k,l$, the question is 
whether 
there is a partition of $V$ into independent sets $V_1,\ldots,V_k \subset V$ such that $|V_i| \leq l$ for each $i=1,\ldots,k$. We shall call this the PIBIS$(G,k,l)$ problem. Since 
PIBIS$(G,k,n)$ is a well-known NP-complete $k$-coloring problem, so is PIBIS$(G,k,l)$, $l<n$. Bodlaender and Jansen \cite{bodlaender} proved that the PIBIS$(G,3,l)$ problem remains 
NP-complete even if $G$ is bipartite. Now we are ready to prove 
\begin{theorem}
The $Q4|p_i=1,G=bipartite|C_{\max}$ problem is \emph{NP}-hard even if $s_1=s_2=s_3$. \label{tw21}
\end{theorem}
\begin{proof}
We prove by reduction from the PIBIS$(G,3,l)$ problem. Suppose we have an instance of PIBIS$(G,3,l)$, i.e. we have a bipartite graph $G$ and we want to know whether there exists a partition 
of its vertices into three independent sets, each of size $\leq l$. We construct the following instance of a scheduling decision problem: machine speeds for $M_1, M_2,$ and $M_3$ are 
$s_1=s_2=s_3=1$. Machine $M_4$ is of speed $s_4 \ll 1/n$ and the limit on schedule length is $l$. The question is whether there exists a schedule of length at most $l$. The membership of 
this problem in class NP is obvious.

The existence of a schedule of length $\leq l$ implies the existence of a 3-partition of $G$ into independent sets of size at most $l$, since no job can be allocated to $M_4$.

If $G$ has a 3-coloring with at most $l$ vertices in each color then our scheduling problem has clearly a solution of length at most $l$, since each color class can be regarded as a load 
on some $M_i$, $i \leq 3$.

The NP-hardness of $Q4|p_i=1,G=bipartite|C_{\max}$ follows from the fact that its decision version is NP-complete.
\end{proof}

\section{Algorithm for the case $s_2=s_3=s_4$}\label{alg}
Since our scheduling problem is NP-hard, we have to propose an approximation algorithm for it. First of all notice that if all the machines are identical then the scheduling problem 
becomes trivial since any equitable 4-coloring of $G$ solves the problem to optimality. Therefore we assume herein that $s_1 \gg s_2=s_3=s_4$ and the incompatibility graph $G$ is of degree 
at most 4.

In the following we will need the concept of an ideal schedule. Let $s = s_1+s_2+s_3+s_4$. A schedule in which all the machines finish at the same time is said to be \emph{ideal}. Note that the 
length of the ideal schedule is $n/s$. Since the number of jobs on each machine must be an integer, the ideal schedule need not be optimal. Such a situation is illustrated in Fig~\ref{fig2}.

The general idea behind our heuristics is to find a semi-equitable coloring in which the largest possible independent set in $G$ is allocated to machine $M_1$ and the remaining job vertices 
are spread equitably within machines $M_2, M_3,$ and $M_4$. This leads to the following \texttt{Algorithm 1} for optimal/suboptimal scheduling in this case.

\begin{algorithm}
\caption{Scheduling in case $s_2=s_3=s_4$}
\begin{algorithmic}
\Require {$n$-vertex graph $G$ with $\Delta(G) \leq 4$ and machine speeds $s_1\gg s_2=s_3=s_4$.}
\Ensure {Optimal/suboptimal schedule.}
\begin{enumerate}
\item Find a maximum independent set $I_1$ in $G$.
\item If $I_1$ contains 6 vertices from the neighborhood of $K_{3,3}$ then set $I_1 = I_1 \cup \{v\} - \{u\}$ as shown in Fig.~\ref{fig3}.
\item Assign $M_1 \leftarrow I_1$.
\item Let $G_1,G_2,\ldots,G_k$ be the set of connected components of $G - I_1$.
\item For each $i = 1,\ldots,k$ find an equitable $(A_i,B_i,C_i)$-coloring of $G_i$, where $|A_i| \leq |B_i| \leq |C_i| \leq |A_i|+1$.
\item Combine the corresponding independent sets to get an equitable coloring $(X,B,Y)$ of the whole $G - I_1$, where $B = B_1\cup B_2\cup \cdots \cup B_k$.
\item Assign $M_2 \leftarrow X$, $M_3 \leftarrow B$, $M_4 \leftarrow Y$.
\end{enumerate}
\end{algorithmic}
\end{algorithm}

\begin{figure}
\begin{center}
\includegraphics{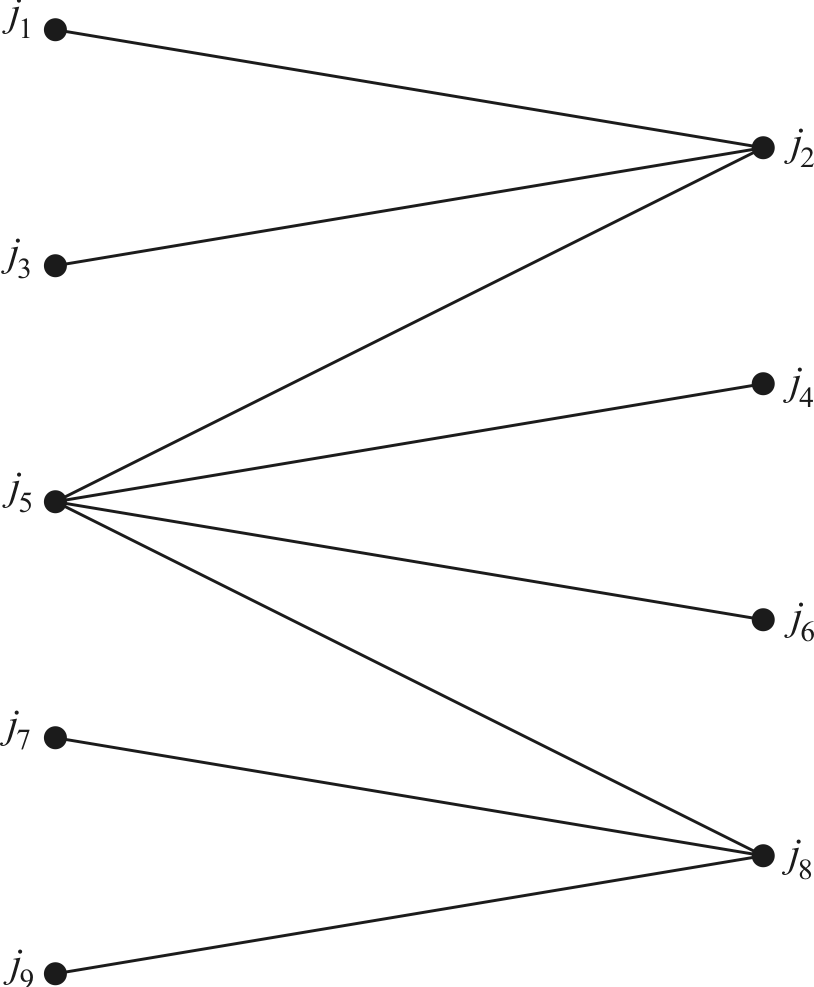} 
\caption{Example of incompatibility graph.} \label{fig1}
\end{center}
\end{figure}

\begin{figure}
\begin{center}
\includegraphics[scale=0.9]{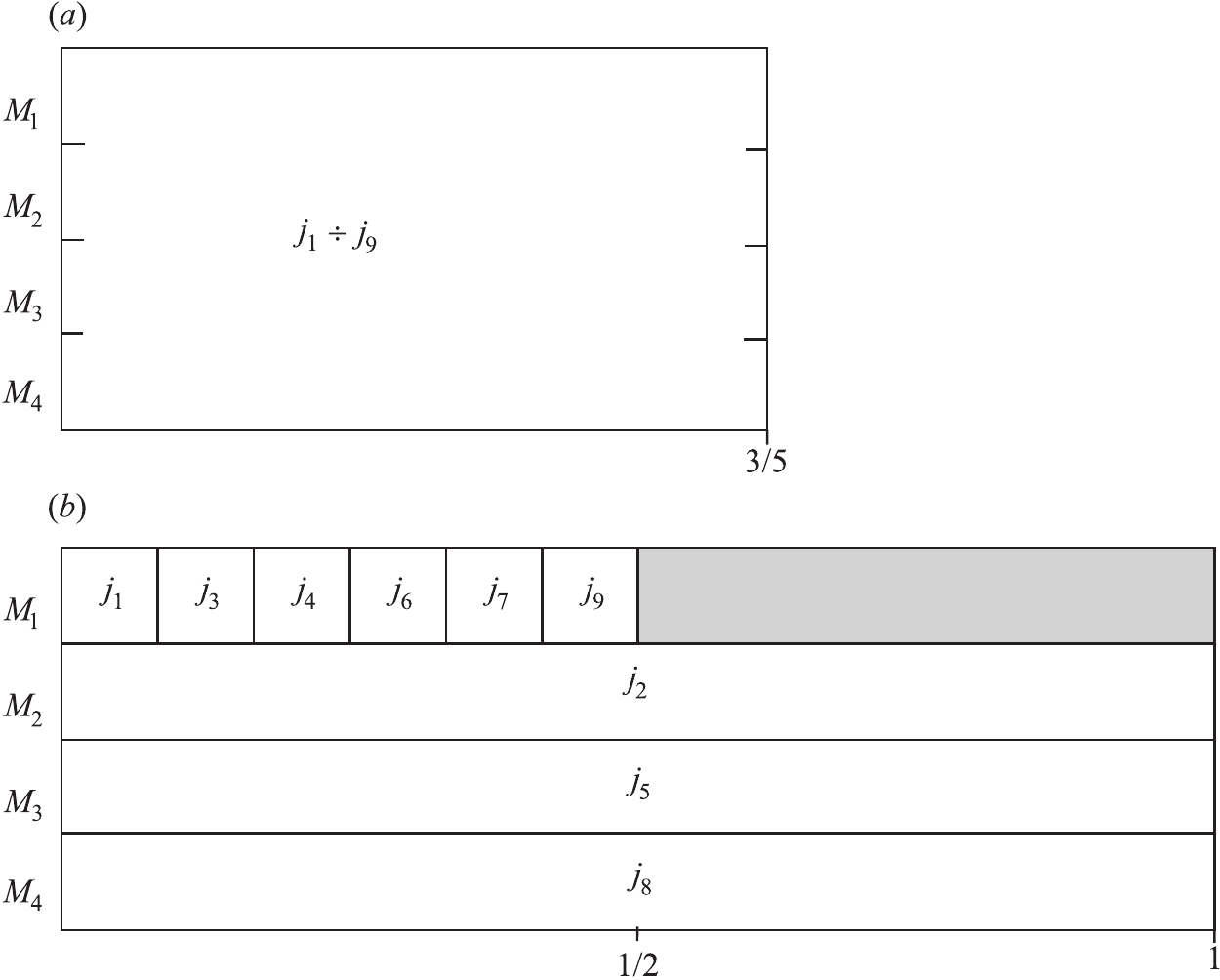} 
\caption{Gantt chart of a schedule for graph of Fig.~\ref{fig1} when $s_1=12$, $s_2=s_3=s_4=1$: (a) ideal; (b) optimal.} \label{fig2}
\end{center}
\end{figure}

The most time-consuming Step 1 of \texttt{Algorithm 1} can be done in $O(n^{1.5})$ time \cite{hop}. 

\begin{figure}
\begin{center}
\includegraphics{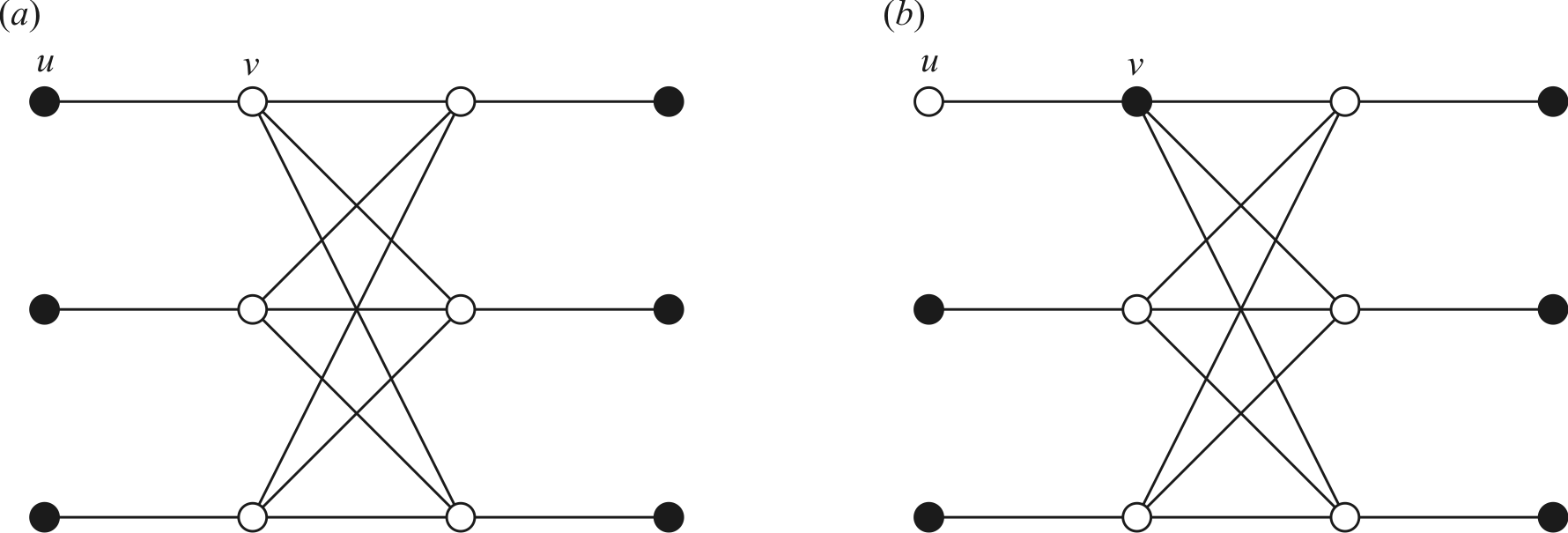} 
\caption{$K_{3,3}$ and its neighborhood (vertices in black belong to $I_1$): (a) before interchange; (b) after interchange.} \label{fig3}
\end{center}
\end{figure}

\begin{theorem}
If $s_1 \geq 12s_2$ and $s_2=s_3=s_4$ then \emph{\texttt{Algorithm 1}} returns an optimal solution. \label{th2}
\end{theorem}
\begin{proof}
Let $G$ be an $n$-vertex incompatibility graph $G$ of degree $\Delta \leq 4$ and let $I_1$ be a maximum cardinality independent set in $G$. If $I_1$ contains 6 vertices from  the 
neighborhood of $K_{3,3}$ like in Fig.~\ref{fig3} then interchange one of them, say $u$, with its neighbor $v$ belonging to $K_{3,3}$, i.e. set $I_1 := I_1 \cup \{v\}  - \{u\}$. First, we will show that $G  -  I_1$ is of 
degree at most 3. If $\Delta(G) \leq 3$, there is nothing to prove.  So suppose that $\Delta(G) = 4$ and let $v$ be any vertex of degree 4. If each such $v \in I_1$ then $\Delta(G  -  I_1) \leq 3$. If $v \not \in I_1$ then at 
least one of its neighbors, say $u$, belongs to $I_1$ since otherwise edge $\{u,v\}$ would not be covered by a minimal vertex cover and, due to K\"{o}nig's theorem 
\cite{konig}, $I_1$ would not be maximal.

Let $C(M_i)$ be the schedule length on machine $M_i$, $i = 1,\ldots,4$. Without loss of generality we may assume that $s_1= 12s_2$ . Then the ideal schedule is of length $n/s = \frac{1}{15}n/s_2$. By inequality 
(\ref{ine}) it follows that $|I_1|\leq \frac{4}{5}n$. Thus $C(M_1) \leq \frac{4}{5}n/(12s_2) = \frac{1}{15}n/s_2$, which is equal to the ideal schedule length. Since $G - I_1$ is a collection of subcubic bipartite graphs different from $K_{3,3}$, 
we can find an optimal scheduling on $M_2, M_3$, and $M_4$ by equitable 3-coloring of $G  - I_1$. In this way the remaining jobs are spread evenly among the three machines $M_i$, $i \geq 2$ which gives 
$\max\{C(M_i): i = 2,3,4\} = C^*_{\max}$, because one cannot do better by moving a job from $M_i$ to $M_1$ as $I_1$ is maximal. This completes the proof of Theorem \ref{th2}.
\end{proof}

\begin{corollary}
If $s_1 \geq 2s_2$ and $s_2=s_3=s_4$ then \emph{\texttt{Algorithm 1}} returns a solution of value at most 2 times $C^*_{\max}$. 
\end{corollary}
\begin{proof}
Without loss of generality we may assume that $s_1 = 2s_2$. In this case the ideal schedule length is $n/s = n/(5s_2) = \frac{2}{5}n/s_1$. Let us consider two extremal 
cases given in inequality (\ref{ine}). 
\begin{description}
\item[\textnormal{\emph{Case} 1:}] $|I_1| =\lceil n/2 \rceil$.

\texttt{Algorithm 1} returns a solution on $M_1$ of length $C(M_1) =\lceil n/2 \rceil /s_1$ which is less than $\frac{4}{5}n/s_1$, i.e. twice the ideal 
schedule length. The remaining jobs are spread evenly among three machines $M_i$, $i \geq 2$, which gives $C(M_i) \cong \lfloor n/2 \rfloor /(s_2+s_3+s_4) =
\lfloor n/2 \rfloor /(1.5s_1) = \frac{2}{3} \lfloor n/2 \rfloor/s_1$ which is less than $\frac{4}{5}n/s_1$, i.e. twice the ideal schedule length. 
The thesis holds in Case 1.

\item[\textnormal{\emph{Case} 2:}] $|I_1| =\lfloor 4n/5 \rfloor$.

Then $C(M_1) =\lfloor 4n/5\rfloor /s_1$ which is  less than $\frac{4}{5}n/s_1$, i.e. twice the ideal schedule length. The remaining jobs are spread evenly among 
$M_i$, $i \geq 2$, which gives $C(M_i) \cong \lfloor n/5 \rfloor/(3s_2) = 2\lfloor n/5 \rfloor /(1.5s_1) = \frac{2}{3} \lfloor n/5\rfloor /s_1 < \frac{4}{3}\lfloor 
n/5\rfloor /s_1 < 2C(M_1)$. The thesis holds in Case 2.
\end{description}

The reader can check that the thesis holds in the remaining cases as well.	
\end{proof}

The worst-case instance for \texttt{Algorithm 1} when $s_1 = 2s_2$ and $G =3K_{1,4}$ is shown in Fig.~\ref{fig5}.

\begin{figure}
\begin{center}
\includegraphics{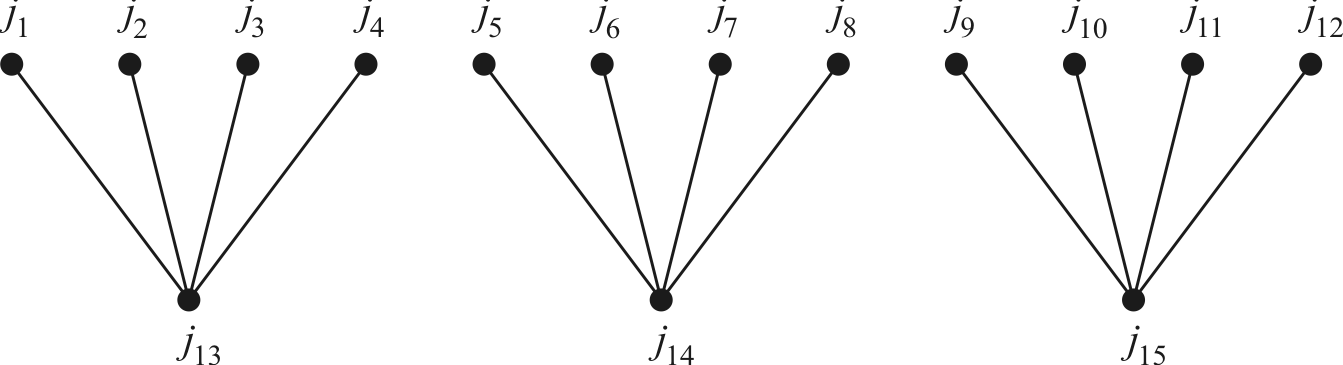} 
\vspace{0.8cm}
\caption{Example of incompatibility graph $G = 3K_{1,4}$.} \label{fig4}
\end{center}
\end{figure}

\begin{figure}
\begin{center}
\includegraphics{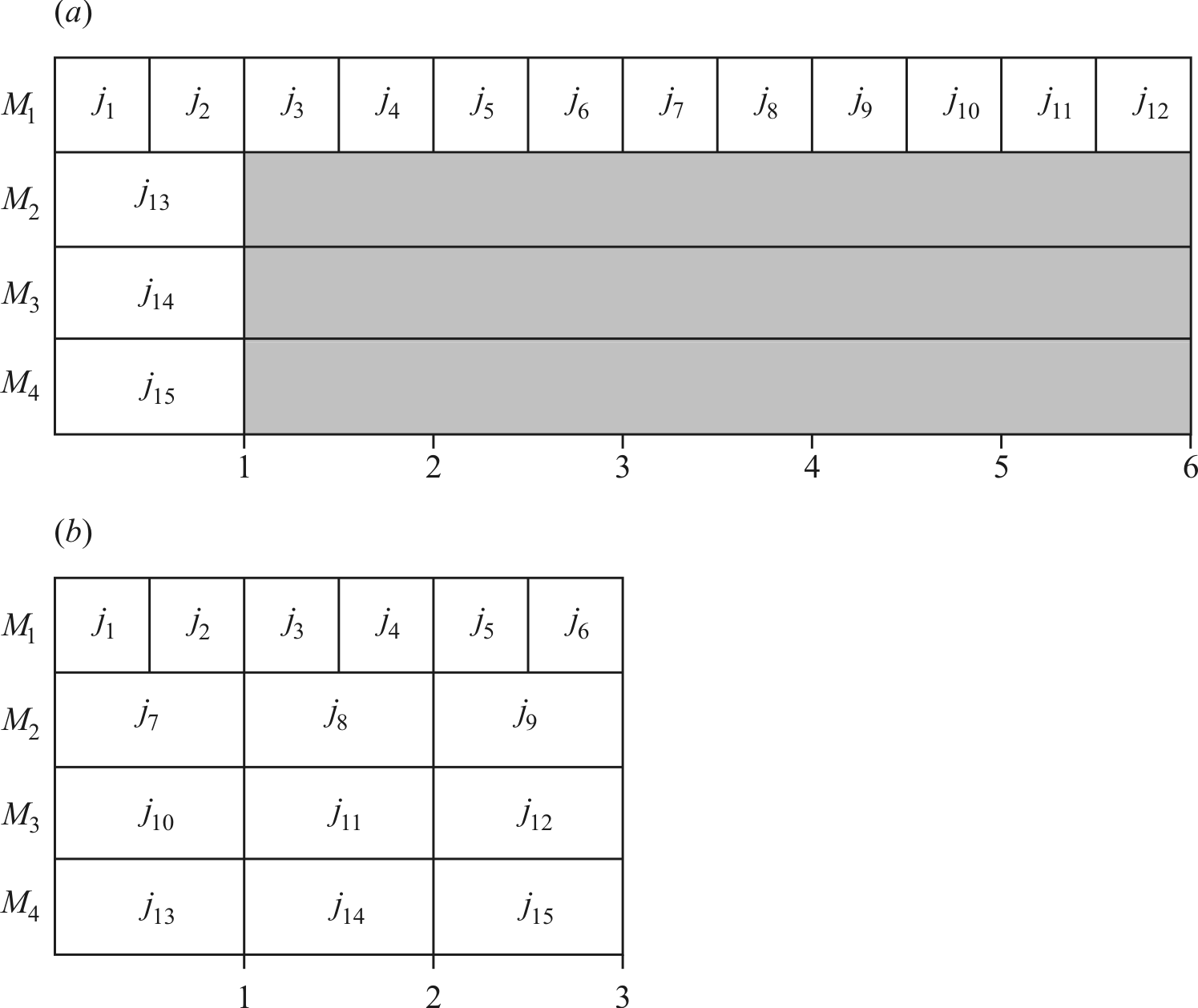} 
\caption{Gantt chart of a schedule for \texttt{Algorithm 1} with graph of Fig.~\ref{fig4} when $s_1 = 2s_i$, $2 \leq i \leq 4$ (a) worst-case schedule; (b) optimal schedule.} \label{fig5}
\end{center}
\end{figure}

\section{Algorithms for the case $s_3=s_4$}\label{alg34}

In this case we have two fast machines $M_1, M_2$,  and two slow machines $M_3, M_4$. Therefore it is reasonable to apply a maximum independent set algorithm 
twice: first towards $G$, which results in a set $I_1$, and then towards $G - I_1$. This idea leads us to an approximation \texttt{Algorithm 2}.

\begin{algorithm}
\caption{Scheduling in case $s_3=s_4$}
\begin{algorithmic}
\Require {$n$-vertex graph $G$ with $\Delta(G) \leq 4$ and machine speeds $s_1 \geq s_2 \geq s_3=s_4$.}
\Ensure {Suboptimal schedule.}
\begin{enumerate}
\item Find $I_1$ a maximum independent set in $G$.
\item Find $I_2$ a maximum independent set in $G - I_1$.
\item Find an equitable $(C,D)$-coloring of $G - I_1 - I_2$.
\item Assign $M_1 \leftarrow I_1, M_2\leftarrow I_2, M_3\leftarrow C, M_4\leftarrow D$.
\end{enumerate}
\end{algorithmic}
\end{algorithm}

\begin{lemma}
If $s_1\geq s_2 \geq 3s_3=3s_4$ then \emph{\texttt{Algorithm 2}} runs in time $O(n^{1.5})$ to find a solution of value at most $32/15$ times $C^*_{\max}$. 
\end{lemma}
\begin{proof}
The complexity of \texttt{Algorithm 2} is obvious since both Steps 1 and 2 can be done in time $O(n^{1.5})$ \cite{hop} while Step 3 is linear.

Let us consider the accuracy of \texttt{Algorithm 2}. For this reason we may assume, without loss of generality, that $s_1=s_2=3s_3$. In this case the ideal schedule length 
is $n/(8s_3)$. Note that the subgraph $G - I_1$ is of degree at most 3. This subgraph may be connected or disconnected. Its order is between $\lfloor n/5\rfloor$ and $\lceil n/2 \rceil$.
Since it is bipartite, we have $\lceil n/10\rceil \leq \alpha(G - I_1) \leq \lfloor n/2 \rfloor$.  Let $I_2$ be a maximum cardinality independent set in $G - I_1$. 
The subgraph $G - I_1 - I_2$ can be equitably colored with 2 colors, since $\Delta(G-I_1-I_2) \leq 2$. In this way we can get a 4-coloring 
ranging from $[\lfloor 4n/5\rfloor, \lceil n/5\rceil, 0, 0]$ to $[\lceil n/2\rceil, \lceil n/4 \rceil, \lceil n/8 \rceil, \lfloor n/8 \rfloor]$. Therefore the schedule length is at most $\lfloor 4n/5 \rfloor / s_1$. Thus
$$
\frac{\text{Alg}_2}{C^*_{\max}} \leq \frac{\lfloor 4n/5 \rfloor /s_1}{n/(8s_3)} \leq \frac{4n/(5s_1)}{n/(8s_3)}=\frac{4/15}{1/8} = \frac{32}{15}.
$$
\end{proof}

In contrast to \texttt{Algorithm 1}, which guarantees an optimal solution to our scheduling problem if $s_1 \geq 12s_2$, no such a guarantee exists for \texttt{Algorithm 2}. In other words, 
there is no bound on $s_2/s_3$ which guarantees that \texttt{Algorithm 2} solves the problem to optimality. In fact, consider graph $G$ depicted in Fig. \ref{2col} and assume that 
$s_1=s_2$. Algorithm 2 when applied to $G$ finds a coloring of type $[6, 1, 1, 0]$, which leads to a schedule of length $\max\{6/s_1, 1/s_3\}$. A better coloring 
is $[4, 4, 0, 0]$ which results in the schedule length of value $4/s_2 = 4/s_1 < 6/s_1$, irrespective of $s_3$.

\begin{figure}
\begin{center}
\includegraphics{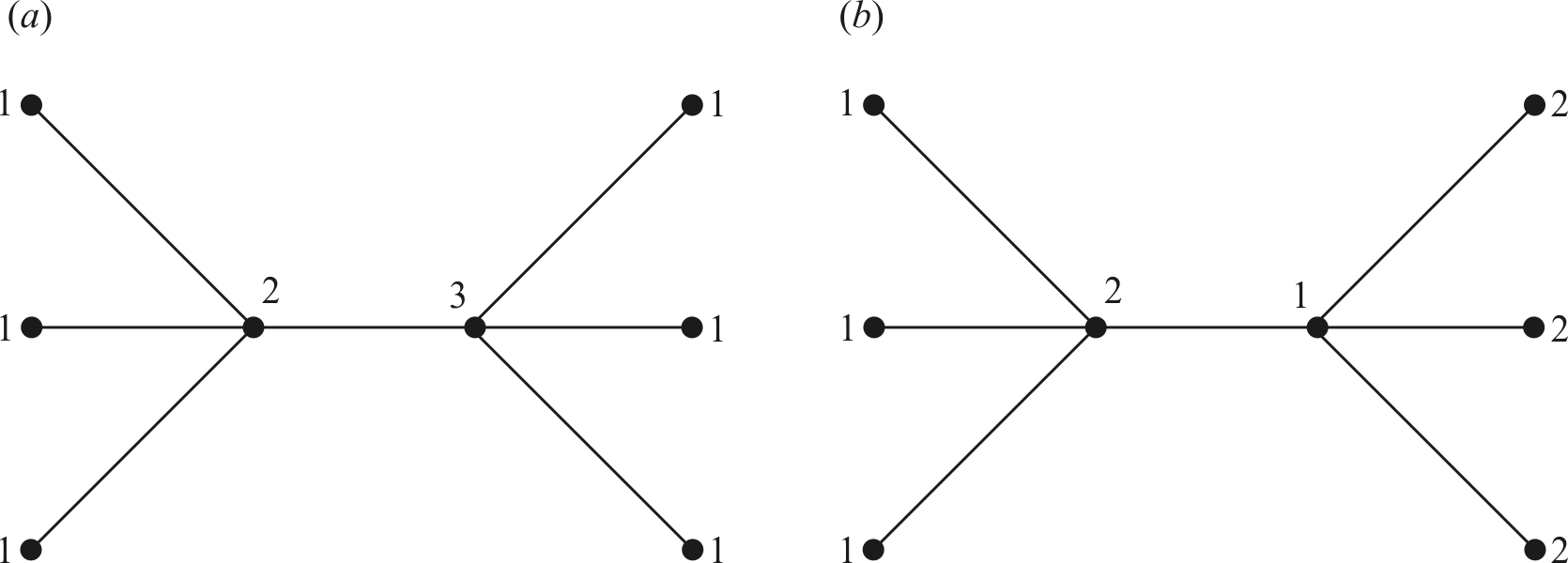} 
\caption{Graph $G$ with two colorings: (a) of type $[6, 1, 1, 0]$; (b) of type $[4, 4, 0, 0]$.} \label{2col}
\end{center}
\end{figure}

Now let us consider an approach based on equitable 4-coloring of $G$. The fact that every bipartite graph of degree $\leq 4$ is equitably 4-colorable was proved by 
Chen and Yen \cite{chen}. In this case we get a coloring of type $[\lceil n/4\rceil,\lceil(n - 1)/4\rceil, \lceil(n - 2)/4\rceil, \lceil (n - 3)/4\rceil]$. 
Hence the schedule length is determined by $C(M_3) = \lceil (n - 2)/4 \rceil/s_3$. The algorithm of this kind is presented as \texttt{Algorithm 3} below. 

\begin{algorithm}
\caption{Scheduling in case $s_3=s_4$}
\begin{algorithmic}
\Require {$n$-vertex graph $G$ with $\Delta(G) \leq 4$ and machine speeds $s_1 \geq s_2 \geq s_3=s_4$.}
\Ensure {Suboptimal schedule.}
\begin{enumerate}
\item Find an equitable $(A,B,C,D)$-coloring of $G$ by applying a procedure described in \cite{chen}.
\item Order the independent sets so that $|A| \geq |B| \geq |C| \geq |D| \geq |A|-1$.
\item Assign $M_1 \leftarrow A, M_2\leftarrow B, M_3\leftarrow C, M_4\leftarrow D$.
\end{enumerate}
\end{algorithmic}
\end{algorithm}
The worst-case ratio of \texttt{Algorithm 3} is bounded above by
$$\frac{C(M_3)}{C(M_2)} = \frac{\lceil(n - 2)/4\rceil /s_3}{\lceil(n - 1)/4\rceil/s_2} \cong \frac{s_2}{s_3},$$
which can be arbitrarily large if $s_3$ is constant and $s_2$ tends to infinity. We have the following Lemma \ref{lm2}.

\begin{lemma}
If $s_1\geq s_2$ and $s_2 \leq 3s_3=3s_4$ then \emph{\texttt{Algorithm 3}} runs in time $O(n)$ to find a solution of value at most $2C^*_{\max}$.\label{lm2}
\end{lemma}
\begin{proof}
Without loss of generality we may assume that $s_1=s_2=3s_3=3s_4$. In this case the length of ideal schedule is $n/(8s_3)$. In the following we consider two cases 
depending on the parity of $n$.
\begin{description}
\item[\textnormal{\emph{Case} 1:}] $n = 2k$.

We have $C(M_3) = \lceil (n - 2)/4 \rceil /s_3 = \lceil (k - 1)/2 \rceil /s_3 \leq \frac{1}{2}k/s_3 = 2n/(8s_3) \leq 2C^*_{\max}$.
\item[\textnormal{\emph{Case} 2:}] $n = 8k - 1, 8k+1, 8k+3, 8k+5$.

In this case the ideal schedule is not optimal, since its length is not an integer. The example of the schedule for the case $n = 8k - 1$ is shown in Fig.~\ref{fig6}.

If $n= 8k - 1$ then an optimal solution corresponds to a coloring of $G$ of type $[3k, 3k, k, k-1]$. Hence $C(M_3) = \lceil(8k - 3)/4\rceil/s_3 = 2k/s_3 = 2C^*_{\max}$.

If $n = 8k+1$ then an optimal solution corresponds to a coloring of $G$ of type $[3k+1, 3k, k, k]$. That is why 
$C(M_3) = \lceil(8k - 1)/4\rceil/s_3 = 2k/s_3 <2(k+\frac{1}{3})/s_3 = 2C^*_{\max}$.

If $n = 8k+3$ then an optimal solution corresponds to a coloring of $G$ of type $[3k+2, 3k+1, k, k]$. Thus $C(M_3) = \lceil(8k+1)/4\rceil/s_3 = (2k+1)/s_3 <2(k+\frac{2}{3})/s_3 = 2C^*_{\max}$.

If $n = 8k+5$ then an optimal solution corresponds to a coloring of $G$ of type $[3k+3, 3k+2, k, k]$. Therefore $C(M_3) = \lceil(8k+3)/4\rceil/s_3 = (2k+1)/s_3 
<2(k+1)/s_3 = 2C^*_{\max}$.
\end{description}
\end{proof}

\begin{figure}
\begin{center}
\includegraphics{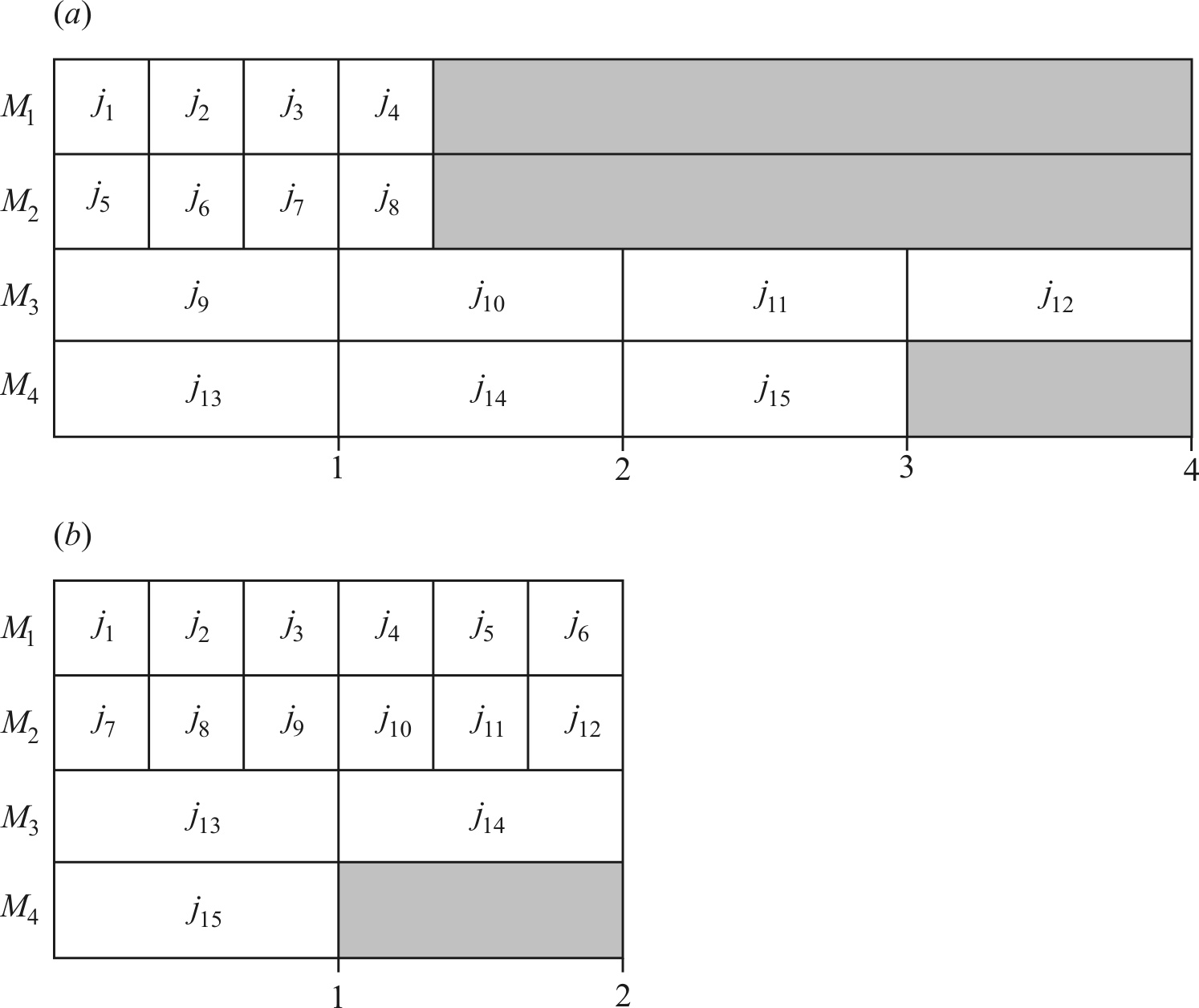} 
\caption{Gantt chart of a schedule for \texttt{Algorithm 3} with graph of Fig.~\ref{fig4} when $s_1=s_2=3s_3=3s_4$: (a) worst-case schedule; (b) optimal schedule.} \label{fig6}
\end{center}
\end{figure}

The above considerations lead us to the following universal algorithm
\begin{algorithm}
\caption{Scheduling in case $s_3=s_4$}
\begin{algorithmic}
\Require {$n$-vertex graph $G$ with $\Delta(G) \leq 4$ and machine speeds $s_1 \geq s_2 \geq s_3=s_4$.}
\Ensure {2-approximate schedule.}
\begin{enumerate}
\item If $s_2>3s_3$ then call \texttt{Algorithm 2} else call \texttt{Algorithm 3}.
\end{enumerate}
\end{algorithmic}
\end{algorithm}

\begin{theorem}
\emph{\texttt{Algorithm 4}} runs in time $O(n^{1.5})$ to produce  a solution of value at most $32/15$ times $C^*_{\max}$. \hfill $\Box$
\end{theorem}

\section{Final remarks}

Our results can be generalized to more than 4 machines. First, suppose that the number of machines $m > 4$. Then the problem $Qm|p_i=1,G=bipartite|C_{\max}$ remains 
NP-hard if 
$s_1=s_2=s_3$. In fact, if $s_4 = \cdots = s_m \ll 1/mn$ then by the same argument as that used in the proof of Theorem \ref{tw21} we get the desired result.

Secondly, the problem $Qm|p_i=1,G=bipartite|C_{\max}$ can be solved to optimality in time $O(n^{1.5})$ by using an algorithm similar to \texttt{Algorithm 1}, if $s_1 \geq m(m+1)s_2$, 
$s_2 = \cdots = s_m$ and $\Delta(G) \leq m$. This is so because under these assumptions the ideal schedule length is $n/s \leq n/(m(m+1)s_2 + s_2 +\cdots+ s_m) = 
n/((m(m+1)s_2) + (m-1)s_2) = n/((m-1)^2s_2)$. 
On the other hand, $C(M_1) \leq mn/((m+1)^2ms_2) = n/((m+1)^2s_2) < n/((m-1)^2s_2)$, i.e. the schedule length on $M_1$ is shorter than the ideal schedule length. Thus an optimal schedule length of the whole system is determined by the 
optimal schedule length on machines $M_2,\ldots, M_m$, since we cannot do better by moving any job from $M_i$ to $M_1$ as the load on the first machine is maximal.

\end{document}